\documentclass[submission,copyright]{eptcs}
\providecommand{\event}{Strategic Reasoning 2013}
\usepackage{amssymb,amsmath,amsthm}

\usepackage{ucs}
\usepackage[utf8x]{inputenc}
\usepackage{url}

\usepackage[noend]{algorithmic}
\usepackage{algorithm}
\usepackage{graphicx}
\usepackage{verbatim}
\usepackage[usenames,dvipsnames]{color}

\usepackage{booktabs}

\usepackage{wrapfig}

\graphicspath{{img/}}

\newcommand{\<}{\langle}
\renewcommand{\>}{\rangle}
\newcommand{\agents}{{\cal A}}

\newcommand{\bigO}[1]{{\cal O}{#1}}
\newcommand{\likes}[2]{l(j,i)}

\newcommand{\lstrat}[1]{\lambda_{#1}}
\newcommand{\strat}{s}
\newcommand{\strats}[1]{strat(#1)}

\newcommand{\vote}[2]{V_{#1}(#2)}
\newcommand{\votep}[1]{P(#1)}
\newcommand{\kvote}[3]{V_{#1}(#2,#3)}
\newcommand{\kvotep}[2]{P(#1,#2)}

\newcommand{\Nat}{{\mathbb N}}

\newcommand{\acro}[1]{\textsc{#1}}
\newcommand{\act}{\mathbb{A}}

\newcommand{\catld}[1]{\langle\!\langle #1\rangle\!\rangle}

\newcommand{\roles}{\mathcal{R}}

\newcommand{\exptime}{$\mathsf{EXPTIME}$}

\newtheorem{definition}{Definition}
\newtheorem{example}{Example}

\newtheorem{theorem}{Theorem}
\newtheorem{lemma}{Lemma}
\newtheorem{proposition}{Proposition}

\title{Concurrent Game Structures with Roles\thanks{A preliminary
    version of this paper was presented during LAMAS workshop held at
    AAMAS on June 4th 2012, and a~talk based on that version was
    given at LBP workshop during ESSLLI summer school, August 2012. It
    is available on arXiv \cite{DyrKazPar12-0}.}}

\author{Truls Pedersen\thanks{Dept. of Information Science and Media
    Studies, University of Bergen, Norway. \url{truls.pedersen@infomedia.uib.no}}
  \and
  Sjur Dyrkolbotn\thanks{Durham Law School, Durham University, UK. \url{s.k.dyrkolbotn@durham.ac.uk}}
  \and
  Piotr Kaźmierczak\thanks{Dept. of Computing, Mathematics and
    Physics, Bergen University College, Norway. \url{phk@hib.no}}
  \and
  Erik Parmann\thanks{Dept. of Informatics, University Bergen,
    Norway. \url{erik.parmann@ii.uib.no}}
}
\def\titlerunning{Concurrent Game Structures with Roles}
\def\authorrunning{T. Pedersen, S. Dyrkolbotn, P. Kaźmierczak \& E. Parmann}

\begin{document}
\maketitle

\begin{abstract}
  In the following paper we present a~new semantics for the well-known
  strategic logic \acro{atl}. It is based on adding \emph{roles} to
  concurrent game structures, that is at every state, each agent
  belongs to exactly one role, and the role specifies what actions are
  available to him at that state. We show advantages of the new
  semantics, provide motivating examples based on sensor networks,
  and analyze model checking complexity.
\end{abstract}
 
\section{Introduction}

\acro{atl} \cite{alur2002alternating} is not only a~highly-expressive and
powerful strategic logic, but also has a~relatively low (polynomial) model
checking complexity. However, as investigated by Jamroga and Dix
\cite{JamDix05-0}, in order for the complexity to be polynomial, the number of
agents must be \emph{fixed}. If the number of agents is taken as a~parameter,
model checking \acro{atl} is $\Delta^{\mathsf{P}}_{2}$-complete or
$\Delta^{\mathsf{P}}_{3}$-complete depending on model representation
\cite{Laroussinie07onthe}.
Also, van der Hoek, Lomuscio and Wooldridge show in \cite{van2006complexity} that the
complexity of model checking is polynomial only if an \emph{explicit}
enumeration of \emph{all} components of the model is assumed. For models
represented in \emph{reactive modules language} (\acro{rml}) complexity of model
checking for \acro{atl} becomes as hard as the satisfiability problem for this
logic, namely \exptime~\cite{van2006complexity}.

We present an alternative semantics that interprets formulas of ordinary
\acro{atl} over concurrent game structures with \emph{roles}. Such structures
introduce an extra element -- a set $R$ of roles and associates each agent with
exactly one role which are considered \emph{homogeneous} in the sense that all 
consequences of the actions of the agents belonging to the topical role is 
captured by considering only the number of ``votes'' an action gets (one vote 
per agent).

We present the revised formalism for \acro{atl} in Section \ref{sec:ratl},
discuss model checking results in Section \ref{sec:mc} and conclude in Section
\ref{sec:concl}.

\section{Role-based semantics for ATL}
\label{sec:ratl}
In this section we will introduce \emph{concurrent game structures
  with roles} (\acro{rcgs}), illustrate them with an example and show in Theorem 
\ref{thm:basicthm} that treating \acro{rcgs} or \acro{cgs} as the semantics 
of \acro{atl} are equivalent.

We will very 
often refer to sets of natural numbers from $1$ to some number $n \geq 1$. To 
simplify the reference to such sets we introduce the notation $[n] = \{1,\ldots,n\}$. 
Furthermore we will let $A^B$
denote the set of functions from $B$ to $A$. We will often also work with
tuples $v = \<v_1,\ldots,v_n\>$ and view $v$ as a
function with domain $[n]$ and write $v(i)$ for $v_i$. Given a
function $f: A \times B \to C$ and $a \in A$, we will use $f_a$ to
denote the function $B \to C$ defined by $f_a(b) = f(a,b)$ for all $b
\in B$.

\begin{definition}
  An \acro{rcgs} is a tuple $H = \langle \agents, R, \roles, Q, \Pi,
  \pi, \act, \delta\rangle$ where:
  \begin{itemize}
  \item $\agents$ is a non-empty set of players. In this text we
    assume $\agents = [n]$ for some $n \in \Nat$, and we reserve
    $n$ to mean the number of agents.
  \item $Q$ is the non-empty set of states.
  \item $R$ is a non-empty set of roles. In this text we assume $R
    =[i]$ for some $i \in \Nat$.
  \item $\roles: Q \times \agents \to R$. For a coalition $A$ we write $A_{r,q}$ to denote the
    agents in $A$ which belong to role $r$ at $q$, and notably $\agents_{r,q}$
    are \emph{all} the agents in role $r$ at $q$.
  \item $\Pi$ is a set of propositional letters and $\pi: Q \to
    \wp(\Pi)$ maps each state to the set of propositions true in it.
  \item $\act : Q \times R \to \mathbb N^+$ is the number of available
    actions in a given state for a given role.
  \item For $\mathcal{A}=[n]$, we say that the set of
    \emph{complete} votes for a role $r$ in a state $q$ is $\vote r q
    = \{v_{r,q} \in [n]^{[\act(q,r)]} \mid \sum_{1 \leq a \leq
      \act(q,r)} v_{r,q}(a) = |\mathcal \agents_{r,q}|\}$, the set of
    functions from the available actions to the number of agents
    performing the action. The functions in this set account for the
    actions of \emph{all} the agents. The set of \emph{complete}
    profiles at $q$ is $\votep q = \prod_{r \in R}\vote r q$. For each
    $q \in Q$ we have a transition function at $q$, $\delta_q: \votep
    q \to Q$ defining a partial function $\delta: Q \times \bigcup_{q
      \in Q}\votep q \to Q$ such that for all $q \in Q$, $P \in \votep
    q$, $\delta(q,P) = \delta_q(P)$.
  \end{itemize}
\end{definition}
The following example illustrates how \acro{rcgs} differs from an ordinary concurrent game
structure:

\begin{example}[Sensor networks]
  \label{ex1}
  A~wireless sensor network is
  a~system composed of a~number of (homogeneous) sensors that can be
  triggered by various stimuli. In
  Figure~\ref{fig:1tier} we show a~\emph{1-tier} (i.e.,\ completely
  homogeneous) sensor network with $n$ sensors. There are two states
  in the system with labels corresponding to an \emph{indicator} of
  the network. $\neg p$ stands for \emph{idle} state of the network,
  while $p$ indicates that the network detected
  a~stimulus. In this very simple example we say that $k$ is our
  \emph{threshold}, i.e.\ if at least $k$ number of sensors detect
  something, then $p$. Since all the sensors behave in the same way we say 
  the role of sensors is homogeneous. Hence the
  system can be modeled using only a single role.  This gives us the model
  depicted in Figure~\ref{fig:1tier}. One can easily add another role to the
  model if needed, for example in a scenario with a ``controller''
  who processes the reported signals, or in a $2$-tier network with several
  types of sensors.
  \begin{figure}[ht]
    \begin{center}
      \includegraphics[scale=.4]{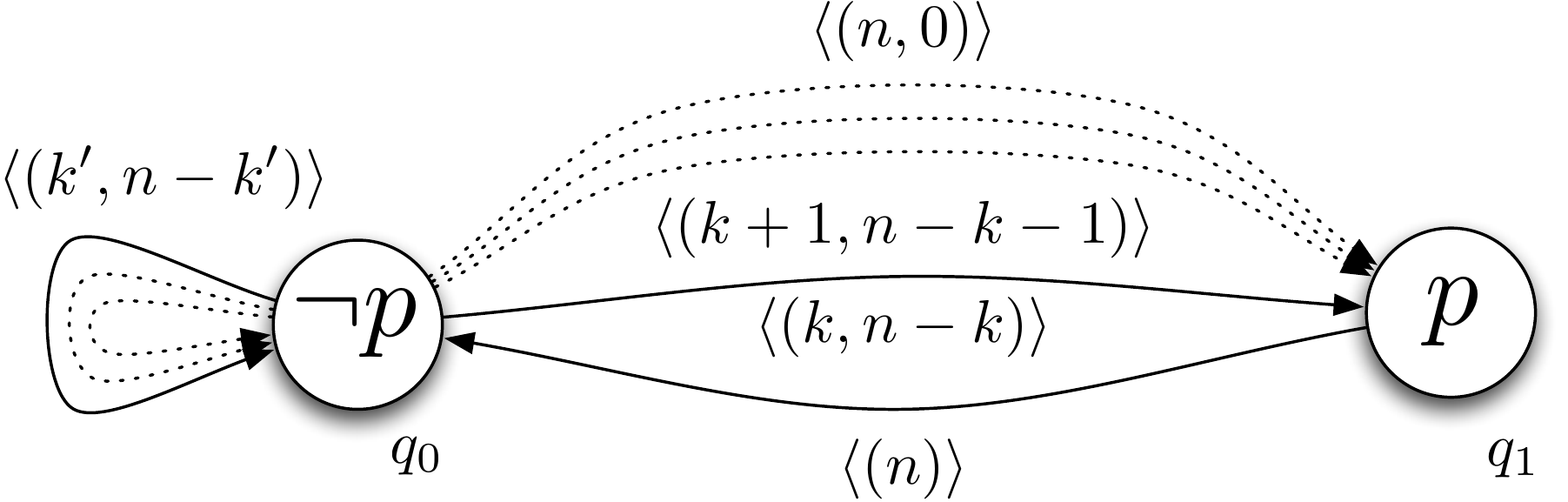}
    \end{center}
    \caption{A depiction of $H_{1}$ -- a~simple 1-tier sensor network.}
    \label{fig:1tier}
  \end{figure}

  A~more complex example is presented in Figure
  \ref{fig:1tiercontroller}, where we add another \emph{role} to our
  structure, that of a~\emph{supervisor} or \emph{controller}. The
  supervisor can act upon sensors' actions, i.e.\ if the sensors report 
  that $p$, the supervisor can perform $q$. As illustrated by the
  drawing, the supervisor has three actions available: he can wait, he
  can reject the message or he can accept the message and proceed to
  state $q_{2}$ performing $q$ (e.g.,\ call the police in an intrusion
  detection scenario).
  \begin{figure}[h]
    \begin{center}
      \includegraphics[scale=.4]{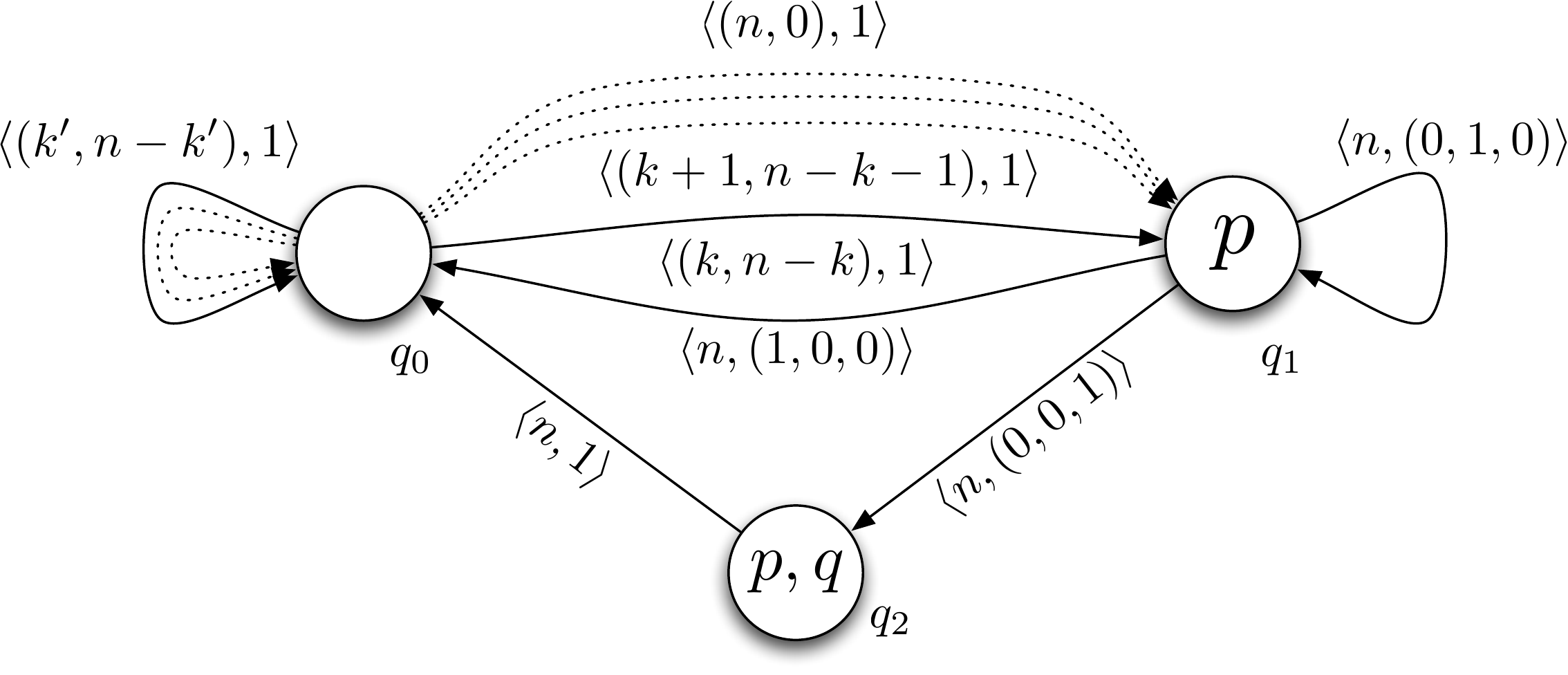}
    \end{center}
    \caption{A sketch of structure $H_{2}$: a~1-tier sensor network with a~supervisor.}
    \label{fig:1tiercontroller}
  \end{figure}
  Finally, in Figure \ref{fig:multitier} we sketch a~\emph{multi-tier}
  example, with two different types of sensors, $n_{1}$ and $n_{2}$,
  each type with its corresponding role. The transition function with
  the addition of a~new role looks like this:
  \[
  \delta_{q_{0}} (\langle (x_{1},
  n_{1}-x_{1}),(x_{2},n_{2}-x_{2})\rangle) = \begin{cases} q_{1}, &
    x_{1} \geq t_{1} \land x_{2} \geq t_{2}\\ q_{0}, &
    \text{otherwise}\end{cases}
  \]
  \[
  \delta_{q_{1}} (\langle n_{1},n_{2}\rangle) = q_{0}
  \]
  where $t_{1}$,$t_{2}$ are thresholds set according to
  \emph{significance} of the sensors.
  \begin{figure}[h]
    \begin{center}
      \includegraphics[scale=.4]{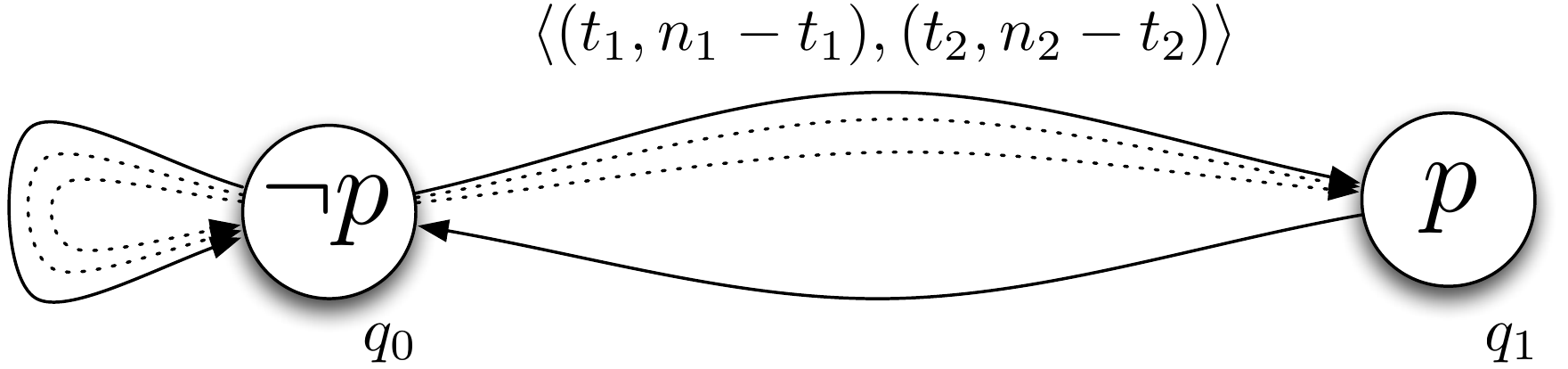}
    \end{center}
    \caption{A sketch of structure $H_{3}$: a~multi-tier sensor network example.}
    \label{fig:multitier}
  \end{figure}

  These simple structures show the benefit of using roles when
  modelling scenarios which involve a~high degree of homogeneity among
  agents. In this simplified sensor setting a~sensor either signals
  that he has made a~relevant observation or he does not -- a binary
  choice. If modelled using concurrent game structures \emph{without}
  roles, models would have $2^{n}$ number of possible action profiles
  in state $q_{0}$, since the identity of the agents signaling that
  they have made an observation has to be accounted for. This,
  however, is irrelevant for the high-level protocol -- all
  that matters is how many sensors of a~given type signal that they
  have made an observation. With roles we can exploit this, and
  we only need to account for the genuinely different scenarios that
  can occur -- corresponding to the number of sensors of each type
  that decide to signal that they have made an observation. In the
  case of just a~single role, this means that we get $n$ as opposed to
  $2^n$ number of different profiles, and the size of the model goes
  from exponential to linear in the number of sensors. In general, as
  we will show in Section \ref{sec:mc}, we shift the exponential
  dependence in the size of models from the number of agents to the
  number of roles.
\end{example}

Given a role $r$, a state $q$ and a coalition $A$, the set of $A$-votes for
$r$ at $q$ is $\kvote r q A$, defined as: $$\kvote r q A = \left\lbrace v \in
  [|A_{r,q}|]^{[\act(q,r)]} ~\left|~ \sum_{a \in [\act(q,r)]} v(a) = |A_{r,q}|\right.\right\rbrace.$$
The $A$-votes for $r$ at $q$ give the possible ways agents in $A$
that are in role $r$ at $q$ can vote. Given a state $q$ and a
coalition $A$, we define the set of $A$-profiles at $q$:
$$
\kvotep q A = \{\langle v_1, \ldots , v_{|R|}\rangle \mid 1 \leq i
\leq |R| : v_i \in \kvote r q A\}.
$$
For any $v \in \kvote r q A$ and $w \in \kvote r q B$ we
write $v \leq w$ iff for all $i \in [\act(q,r)]$ we have $v(i) \leq
w(i)$. If $v \leq w$, we say that $w$ \emph{extends} $v$. If $F =
\langle v_1, \ldots, v_R \rangle \in \kvotep q A$ and $F' = \langle
v_1', \ldots, v_R'\rangle \in \kvotep q B$ with $v_i \leq v_i'$ for
every $1 \leq i \leq |R|$, we say that $F \leq F'$ and that $F$
extends $F'$.
Given a (partial) profile $F'$ at a state $q$ we
write $ext(q,F)$ for the set of all complete profiles that extend
$F'$.

Given two states $q,q' \in Q$, we say that $q'$ is a \emph{successor}
of $q$ if there is some $F \in \votep q$ such that $\delta(q,F) =
q'$. A \emph{computation} is an infinite sequence $\lambda =
q_0q_1\ldots$ of states such that for all positions $i\geq 0$,
$q_{i+1}$ is a successor of $q_i$. We follow standard
abbreviations, hence a $q$-computation denotes a computation starting at
$q$, and $\lambda[i]$, $\lambda[0,i]$ and $\lambda[i,\infty]$ denote
the $i$-th state, the finite prefix $q_0q_1\ldots q_i$ and the
infinite suffix $q_iq_{i+1}\ldots$ of $\lambda$ for any computation
$\lambda$ and its position $i\geq 0$, respectively. An \emph{$A$-strategy} for $A
\subseteq \agents$ is a function $s_A: Q \to \bigcup_{q \in Q}\kvotep
q A$ such that $s_A(q) \in \kvotep q A$ for all $q \in Q$. That is,
$s_A$ maps states to $A$-profiles at that state. The set of all
$A$-strategies is denoted by $\strats A$. 
When needed to distinguish between different
structures we write $\strats {S,A}$ to indicate that we are talking
about the set of strategies for $A$ in structure $S$.
If $s$ is an
$\agents$-strategy and we apply $\delta_{q}$ to $s(q)$, we obtain a
unique new state $q' = \delta_q(\strat(q))$. Iterating, we get the
\emph{induced} computation $\lambda_{s,q} = q_{0}q_{1}\ldots$ such
that $q = q_0$ and $\forall i~\geq 0: \delta_{q_{i}} (s(q_i)) =
q_{i+1}$. Given two strategies $s$ and $s'$, we say that $s \leq s'$
iff $\forall q \in Q: s(q) \leq s'(q)$. Given an $A$-strategy $s_A$
and a state $q$ we get an associated \emph{set} of computations
$out(s_A,q)$. This is the set of all computations that can result when
at any state, the players in $A$ are voting/acting in
the way specified by $s_A$, that is $out(s_A,q) = \{\lstrat {\strat,q} \mid \strat \textit{ is an}\
\mathcal{A}\textit{-strategy and } \strat \geq s_A\}.$ 


Given the definitions above, we can interpret \acro{atl} formulas in the following manner, leaving
out the propositional cases and abbreviations: 
\begin{definition}\label{def:basicsem} 
  Given a \acro{rcgs} $S$ and a state $q$ in $S$, we define the
  satisfaction relation $\models$ inductively:
  \begin{itemize}
  \item $S,q \models \catld{A}\bigcirc\phi$ iff there is $s_A \in
    \strats A$ such that for all $\lambda \in out(s_A,q)$, we have $S,
    \lambda[1] \models \phi$
    \item $S,q \models \catld{A}\phi\mathcal{U}\phi'$ iff there is $s_A
    \in \strats A$ such that for all $\lambda \in out(s_A,q)$ we have
    $S, \lambda[i] \models \phi'$ and $S, \lambda[j] \models \phi$ for
    some $i \geq 0$ and for all $0 \leq j < i$
  \end{itemize}
\end{definition}

Towards the statement that interpreting formulas over \acro{cgs} and \acro{rcgs} 
is equivalent (Theorem \ref{thm:basicthm}) we will describe a surjective translation 
function $f$ translating each \acro{rcgs} to a \acro{cgs}. The following two lemmas 
will be useful in formulating the proof of Theorem \ref{thm:basicthm}.

The translation function $f$ from \acro{rcgs} to \acro{cgs} is defined
as follows:

$$ f \langle \agents, R, \roles, Q, \Pi, \pi, \act, \delta\rangle = \langle \agents, Q, \Pi, \pi, d, \delta' \rangle  $$
where:
\begin{align*}
  d_a(q) &= \act(q, r) & \text{where } a \in \roles(q, r) \\
  \delta'(q, \alpha_1, \ldots, \alpha_n) &= \delta(q, v_1, \dots, v_{|R|}) &\text{where for each role $r$} \\
\end{align*}
\vspace{-3.28em}
\[
v_r = \langle |\{i \in \roles(q, r)~|~ \alpha_{i} = 1\}|, \ldots ,
|\{i \in \roles(q, r)~|~ \alpha_{i} = \act(q,r)\}| \rangle
\]

We describe a surjective function $m : \strats {f(S)} \to \strats S$
mapping action tuples and strategies of $f(S)$ to profiles and
strategies of $S$ respectively. For all $A \subseteq \agents$ and any
action tuple for $A$ at $q$, $t_q =
\<\alpha_{a_1},\alpha_{a_2},...,\alpha_{a_{|A|}}\>$ with $1 \leq
\alpha_{a_i} \leq d_{a_i}(q)$ for all $1 \leq i \leq |A|$, the
$A$-profile $m(t_q)$ is defined in the following way:
\begin{align*}
  m(t_q) &= \<v(t_q,1),\ldots,v(t_q,|R|)\> \text{ where for all } 1 \leq r \leq |R| \text{ we have } \\
  v(t_q,r) &= \<|\{a \in A_{r,q} \mid \alpha_a = 1\}|,\ldots,|\{a \in
  A_{r,q} \mid \alpha_a =
  \act(q,r)\}|\> 
\end{align*}
\begin{lemma}
  \label{lemma:surj}
  For any \acro{rcgs} $S$ and any $A \subseteq \agents$, the function $m:
  strat(f(S), A) \to strat(S, A)$ is surjective.
\end{lemma}

\begin{proof}
  Let $p_A$ be some strategy for $A$ in $S$. We must show there is a
  strategy $s_A$ in $f(S)$ such that $m(s_A) = p_A$.  For all $q \in
  Q$, we must define $s_A(q)$ appropriately. Consider the profile
  $p_A(q) = \<v_1,\ldots,v_{|R|}\>$ and note that by definition of a
  profile, all $v_r$ for $1 \leq r \leq |R|$ are $A$-votes for $r$ and
  that by definition of an $A$-vote, we have $\sum_{1 \leq i \leq
    \act(q,r)}v_r(i) = |A_{r,q}|$. Also, for all agents $a,a' \in
  A_{r,q}$ we know, by definition of $f$, that $d_a(q) = d_{a'}(q) =
  \act(q,r)$.

  It follows that there are functions $\alpha: A \to \mathbb
  N^+$ such that for all $a \in A$, $\alpha(a) \in [d_a(q)]$ and $|\{a
  \in A_{r,q} \mid \alpha(a) = i\}| = v_r(i)$ for all $1 \leq i \leq
  \act(q,r)$, i.e.
$$v_r = \< |\{a \in A_{r,q} | \alpha (a) = 1\}|, \dots, 
|\{a \in A_{r,q} | \alpha (a) = \act(q,r)\}| \>$$ We choose some such
$\alpha$ and $s_A = \<\alpha(a_1), \dots, \alpha(a_{|A|})\>$. Having
defined $s_A$ in this way, it is clear that $m(s_A) = p_A$.
\end{proof}

It will be
useful to have access to the set of states that can result in the next
step when $A \subseteq \agents$ follows strategy $s_A$ at state $q$,
$succ(q,s_A) = \{q' \in Q \mid \exists F \in ext(q,s_A): \delta(q,F) =
q'\}$.
Given either a~\acro{cgs} or an \acro{rcgs} $S$, we define the set of
sets of states that a~coalition $A$ can \emph{enforce} in the next
state of the game:
$$force(S,q,A) = \{succ(q,s_A) \mid s_A \text{ is a strategy for } A
\text{ in } S\}.$$

Using the surjective function $m$ we can prove the following lemma,
showing that the ``next time'' strength of any coalition $A$ is the same
in $S$ as it is in $f(S)$.

\begin{lemma}
  \label{lemma:basic}
  For any \acro{rcgs} $S$, and state $q \in Q$ and any coalition $A
  \subseteq \agents$, we have $force(S, A, q) = force(f(S), A, q)$.
\end{lemma}

\begin{proof}
  By definition of $force$ and Lemma \ref{lemma:surj} it is sufficient
  to show that for all $s_A \in \strats{f(S),A}$, we have
  $succ(S,m(s_A),q) = succ(f(S),s_A,q)$. We show $\subseteq$ as
  follows: Assume that $q' \in force(S,m(s_A),q)$. Then there is some
  complete profile $P = \<v_1,\ldots,v_{|R|}\>$, extending
  $m(s_A)(q)$, such that $\delta(q,P) = q'$. Let $m(s_A)(q) =
  \<w_1,\ldots,w_{|R|}\>$ and form $P' = \<v'_1,\ldots,v'_{|R|}\>$
  defined by $v'_i = v_i - w_i$ for all $1 \leq i \leq |R|$. Then each
  $v'_i$ is an $(\agents \setminus A)$-vote for role $i$, meaning that
  the sum of entries in the tuple $v'_i$ is $|(\agents \setminus
  A)_{r,q}|$.  This means that we can define a function $\alpha:
  \agents \to \Nat^+$ such that for all $a \in \agents$, $\alpha(a)
  \in [d_a(q)]$ and for all $a \in A$, $\alpha(a) = s_a(q)$ and for
  every $r \in R$ and every $a \in (\agents \setminus A)$, and every
  $1 \leq j \leq \act(q,r)$, $|\{a \in (\agents \setminus A)_{r,q}
  \mid \alpha(a) = j\}| = v'_r(j)$. Having defined $\alpha$ like this
  it follows by definition of $m$ that for all $1 \leq j \leq
  \act(q,r)$, $|\{a \in A_{r,q} \mid \alpha(a) = j\}| = w_r(j)$. Then
  for all $r \in R$ and all $1 \leq j \leq \act(q,r)$ we have $|\{a
  \in \agents_{q,r} \mid \alpha(a) = j\}| = v_r(j)$. By definition of
  $f(S)$ it follows that $q' = \delta(q,P) = \delta'(q,\alpha)$ so
  that $q' \in force(f(S),s_A,q)$. We conclude that $force(S,f(s_A),q)
  \subseteq force(f(S),s_A,q)$. The direction $\supseteq$ follows
  easily from the definitions of $m$ and $f$.
\end{proof}

We now state and prove the equivalence. 

\begin{theorem}
  \label{thm:basicthm}
  For any \acro{rcgs} $S$, any $\phi$ and any $q \in Q$, we have $S, q
  \models \phi$ iff $f(S), q \models_{\acro{cgs}}\phi$, where $f$ is
  the surjective model-translation function.
\end{theorem}

\begin{proof}
  Given a structure $S$, and a formula $\phi$,
  we define $true(S,\phi) = \{q \in Q \mid S,q \models
  \phi\}$. Equivalence of models $S$ and $f(S)$ is now demonstrated by
  showing that the equivalence in next time strength established in
  Lemma \ref{lemma:basic} suffices to conclude that $true(S,\phi) =
  true(f(S),\phi)$ for all $\phi$.

  We prove the theorem by showing that for all $\phi$, we have
  $true(S,\phi) = true(f(S),\phi)$. We use induction on complexity of
  $\phi$. The base case for atomic formulas and the inductive steps
  for Boolean connectives are trivial, while the case of
  $\catld{A}\bigcirc\phi$ is a straightforward application of Lemma
  \ref{lemma:basic}. For the cases of $\catld{A}\Box\phi$ and
  $\catld{A}\phi\mathcal{U}\psi$ we rely on the following fixed point
  characterizations, which are well-known to hold for \acro{atl}, see
  for instance \cite{Jamroga:2009:EYH:1615285.1615287}, and are also easily
  verified against definition \ref{def:basicsem}:
  \begin{gather}\label{eq:fixed}
    \begin{gathered}
      \catld{A}\Box\phi \leftrightarrow \phi \land
      \catld{A}\bigcirc\catld{A}\Box\phi\\
      \catld{A}\phi_{1}\mathcal{U}\phi_{2} \leftrightarrow
      \phi_{2}\lor
      (\phi_{1}\land\catld{A}\bigcirc\catld{A}\phi_{1}\mathcal{U}\phi_{2}
    \end{gathered}
  \end{gather}
  We show the induction step for $\catld{A}\Box\phi$, taking as
  induction hypothesis $true(S,\phi) = true(f(S),\phi)$.  The first
  equivalence above identifies $Q' = true(S,\catld{A}\Box\phi)$ as the
  maximal subset of $Q$ such that $\phi$ is true at every state in
  $Q'$ and such that $A$ can enforce a state in $Q'$ from every state
  in $Q'$, i.e. such that $\forall q \in Q': \exists Q'' \in
  force(q,A): Q'' \subseteq Q'$.  Notice that a unique such set always
  exists. This is clear since the union of two sets satisfying the two
  requirements will itself satisfy them (possibly the empty set). The
  first requirement, namely that $\phi$ is true at all states in $Q'$,
  holds for $S$ iff if holds for $f(S)$ by induction hypothesis. Lemma
  \ref{lemma:basic} states $force(S,q,A) = force(f(S),q,A)$, and this
  implies that also the second requirement holds in $S$ iff it holds
  in $f(S)$. From this we conclude $true(S,\catld{A}\Box\phi) =
  true(f(S),\catld{A}\Box\phi)$ as desired. The case for
  $\catld{A}\phi\mathcal{U}\psi$ is similar, using the second
  equivalence.
\end{proof}

\begin{example}[Sensor networks contd.]
  \label{ex2}
  To further illustrate the use of \acro{atl} interpreted over
  \acro{rcgs}, we provide example formulas that are related to the
  structures shown in Example \ref{ex1}.

  In the structure depicted in Figure \ref{fig:1tier}, if at least $k$
  sensors signal something, $p$ becomes true (e.g.\ the alarm is
  raised). This is expressed by formula $\langle\!\langle
  A\rangle\!\rangle\bigcirc p$ which is satisfied in the structure
  from Figure \ref{fig:1tier}, i.e.\ $H_{1}, q_{0} \vDash
  \langle\!\langle A\rangle\!\rangle\bigcirc p$ whenever $|A\cap
  \roles (q_{0},1)| \geq k$. In Figure \ref{fig:1tiercontroller}, the
  supervisor decides whether signals that indicate $p$ are strong
  enough in order for \emph{him} to signal $q$, e.g.\ raise the
  alarm. In this scenario, the sensors alone cannot raise the alarm,
  hence $H_{2}, q_{0} \not\models \langle\!\langle A\rangle\!\rangle
  \Diamond q$ whenever $A \cap \roles(q_{1},2) = \emptyset$ (which
  means that whenever the coalition $A$ does not include the
  supervisor, $q$ cannot be enforced). On the other hand, $H_{2},
  q_{0} \models \catld{A}\bigcirc\catld{B}\bigcirc q$ whenever $|A
  \cap \roles(q_{0},1)| \geq k$ and $B \cap \roles(q_{1},2) \neq
  \emptyset$ (which means that the coalition of agents without
  a~supervisor can enable the supervisor to take action).
\end{example}

\section{Model checking and the size of models}
\label{sec:mc}
In this section we will see how using roles can lead to a dramatic decrease in
the size of \acro{atl} models. 
We first investigate the size of models in terms of
the number of roles, players and actions, and then we analyze model
checking of \acro{atl} over concurrent game structures with roles.

Given a set of numbers $[a]$ and a number $n$, it is a well-known
combinatorial fact that the number of ways in which to choose $n$
elements from $[a]$, allowing repetitions, is $\frac{(n +
  (a-1))!}{n!(a-1)!}$. Furthermore, this number satisfies the
following two inequalities:\footnote{If this is not clear, remember
  that $n^a$ and $a^n$ are the number of functions $[n]^{[a]}$ and
  $[a]^{[n]}$ respectively. It should not be hard to see that all ways
  in which to choose $n$ elements from $a$ induce non-intersecting
  sets of functions of both types.}
\begin{equation}\label{eq:basic}
  \begin{array}{ccc}\frac{(n + (a-1))!}{n!(a-1)!} \leq a^n &\text{ and }& \frac{(n + (a-1))!}{n!(a-1)!} \leq n^a.\end{array}
\end{equation}
These two inequalities provide us with an upper bound on the
\emph{size} of \acro{rcgs} models that makes it easy to compare their
sizes to that of \acro{cgs} models.  Typically, the size of concurrent
game structures is dominated by the size of the domain of the
transition function. For an \acro{rcgs} and a given state $q \in Q$
this is the number of complete profiles at $q$. To measure it,
remember that every complete profile is an $|R|$-tuple of votes $v_r$,
one for each role $r \in R$. 
Also remember
that a vote $v_r$ for $r \in R$ is an $\act(q,r)$-tuple such that the
sum of entries is $|\agents_{q,r}|$. Equivalently, the vote $v_r$ can be
seen as the number of ways in which we can make $|\agents_{q,r}|$
choices, allowing repetitions, from a set of $\act(q,r)$
alternatives. Looking at it this way, we obtain:
$$|\votep q| = \prod_{r\in
  R}\frac{(|\agents_{q,r}| + (\act(q,r) -
  1))!}{|\agents_{q,r}|!(\act(q,r)-1))!}.
$$ 
  We sum over all $q \in Q$ to obtain what we consider to be the size of
an \acro{rcgs} $S$. In light of Equation~\ref{eq:basic}, it follows
that the size of $S$ is upper bounded by both of the following
expressions.

\begin{equation}\label{eq:size}
  \begin{array}{ccc} \bigO (\sum_{q\in Q}\prod_{r \in R}|\agents_{q,r}|^{\act(q,r)}) & \text{ and }  & \bigO (\sum_{q \in Q}\prod_{r\in R}\act(q,r)^{|\agents_{q,r}|}).
  \end{array}
\end{equation}
We observe that growth in the size of models is polynomial in $a =
max_{q \in Q,r \in R}\act(r,q)$ if $n = |\agents|$ and $|R|$ is fixed,
and polynomial in $p = max_{q \in Q, r \in R}|\agents_{q,r}|$ if $a$ and
$|R|$ are fixed. This identifies a significant potential advantage
arising from introducing roles to the semantics of \acro{atl}. The
size of a \acro{cgs} $M$, when measured in the same way, replacing
complete profiles at $q$ by complete action tuples at $q$, grows
exponentially in the players whenever the players have more than one action. We stress that we are \emph{not} just counting the number of
transitions in our models differently. We do have an additional
parameter, the roles, but this is a new semantic construct
that gives rise to genuinely different semantic structures. We have
established that it is possible to use them to give the semantics of
\acro{atl}, but this does not mean that there is not more to be said
about them. Particularly crucial is the question of model checking
over \acro{rcgs} models.

\subsection{Model checking using roles}

For \acro{atl} there is
a well known model checking algorithm \cite{alur2002alternating}. It does model checking in time linear
in the length of the formula and the size of the model. 
Given a structure $S$, and a formula $\phi$, the standard
model checking algorithm  $mcheck(S,\phi)$ returns the set of
states of $S$ where $\phi$ holds.

\begin{wrapfigure}{r}{0.4\textwidth}
  \label{alg:force}
    \begin{algorithmic}
      \FOR {$F \in \kvotep q A$} \STATE $p \leftarrow true$ \FOR {$F' \in ext(q,F)$} \IF
      {$\delta(q,F') \not \in Q'$} \STATE $p \leftarrow false$
      \ENDIF
      \ENDFOR
      \IF {$p = true$} \RETURN $true$
      \ENDIF
      \ENDFOR
      \RETURN $false$
      \caption{$enforce(S,A,q,Q')$}
    \end{algorithmic}
  \end{wrapfigure}
The algorithm depends on a~function $enforce(S,A,q,Q')$, which 
given a structure $S$, a coalition $A$, a state $q \in Q$ and a set of
states $Q'$ answers true or false depending on
whether or not $A$ can enforce $Q'$ from $q$. 
This is the only part of the standard algorithm
that needs to be modified to accommodate roles.

For all profiles $F \in \kvotep q A$ the $enforce$ algorithm runs through all
complete profiles $F' \in \votep q$ that extend $F$.
It is polynomial in the number of complete profiles,
since for any coalition $A$ and state $q$ we have $|\kvotep q A| \leq
|\votep q|$, meaning that the complexity of $enforce$ is upper bounded
by $|\votep q|^2$.
Given a fixed length formula and a fixed number of states, $enforce$
dominates the running time of $mcheck$.
It follows that model checking of \acro{atl} over
concurrent game structures with roles is polynomial in the size of the
model. We summarize this result.
\begin{proposition}\label{prop:modelcheck}
  Given a \acro{cgs} $S$ and a formula $\phi$, $mcheck(S,\phi)$ takes
  time $\bigO(l e^2)$ where $l$ is the length of $\phi$ and $e =
  \sum\limits_{q \in Q}|\votep q|$ is the total number of transitions in
  $S$
\end{proposition}Since model checking \acro{atl} over \acro{cgs}s takes only linear
time, $\bigO (l e)$, adding roles apparently makes model checking
harder. On the other hand, the \emph{size} of \acro{cgs} models can be
bigger by an exponential factor, making model checking much easier
after adding roles. In light of the bounds we have on the size of
models, c.f. Equation \ref{eq:size}, we find that as long as the roles
and the actions remain fixed, complexity of model checking is only
polynomial in the number of agents. This is a potentially significant
argument in favor of including roles in the semantics.

Roles should be used at the modeling stage, as they give the modeler
an opportunity for exploiting homogeneity of the system under
consideration. We think that it is reasonable to hypothesize that in
practice, most large scale multi-agent systems that lend themselves well to
modeling by \acro{atl}  exhibit
significant homogeneity.

The question arises as to whether or not using an \acro{rcgs} is
\emph{always} the best choice, or if there are situations when the
losses incurred in the complexity of model checking outweigh the gains
we make in terms of the size of models. We conclude with the following
proposition, also shown in \cite{DyrKazPar12-0}, which states that as
long we use the standard algorithm, model checking any \acro{cgs} $M$
can be done at least as quickly by model checking an \emph{arbitrary}
$S \in f^-(M)$.

\begin{proposition}\label{prop:alwaysbest}
  Given any \acro{cgs}-model $M$ and any formula $\phi$, let
  $c(mcheck(M,\phi))$ denote the complexity of running
  $mcheck(M,\phi)$. We have, for all $S \in f^-(M)$, that complexity
  of running $mcheck(S,\phi)$ is $\bigO (c(mcheck(M,\phi))$
\end{proposition}

\begin{proof}
  It is clear that for any $S \in f^-(M)$, running $mcheck(S,\phi)$
  and $mcheck(M,\phi)$, a difference in overall complexity can arise
  only from a difference in the complexity of $enforce$. So we compare
  the complexity of $enforce(S,A,q,Q'')$ and $enforce(M,A,q,Q'')$ for
  some arbitrary $q \in Q$, $Q'' \subseteq Q$. The complexity in both
  cases involves passing through all complete extensions of all
  strategies for $A$ at $q$. The sizes of these sets can be
  compared as follows, the first inequality is an instance of Equation
  \ref{eq:basic} and the equalities follow from definition of $f$ and
  the fact that $M = f(S)$.
  \begin{align*}
    \prod_{r \in R} \left(\frac{(|A_{r,q}| + (\act(r,q)- 1))!}{|A_{r,q}|!(\act(r,q)-1)!}\right) \times& \prod_{r \in R}\left(\frac{((|\agents_{q,r}| - |A_{r,q}|) + (\act(r,q) - 1))!}{(|\agents_{q,r}| - |A_{r,q}|)!(\act(r,q)-1)!}\right)  \\ \leq& \left(\prod_{r \in R} \act(r,q)^{|A_{r,q}|} \times \prod_{r \in R} \act(r,q)^{|\agents_{q,r}| - |A_{r,q}|}\right) \\
    =& \prod_{r \in R} \left(\prod_{a \in A_{r,q}} \act(r,q)\right) \times \prod_{r\in R} \left(\prod_{a \in \agents_{a,r} \setminus A_{r,q}} \act(r,q)\right) \\
    =& \left(\prod_{a \in A}d_a(q) \times \prod_{a \in \agents
        \setminus
        A}d_a(q)\right) = \prod_{a \in \agents}d_a(q)\\
  \end{align*}
  We started with the number of profiles (transitions) we need to
  inspect when running $enforce$ on $S$ at $q$, and ended with the
  number of action tuples (transitions) we need to inspect when
  running $enforce$ on $M = f(S)$. Since we showed the first to be
  smaller or equal to the latter and the execution of all other
  elements of $mcheck$ are identical between $S$ and $M$, the claim
  follows.
\end{proof}

\section{Conclusions, related and future work}
\label{sec:concl}

In this paper we have described a~new type of semantics for the
strategic logic \acro{atl}. We have provided illustrating examples and
argued that although in principle model checking \acro{atl}
interpreted over concurrent game structures with roles is harder than
the standard approach, it is still polynomial and can generate
exponentially smaller models. We believe this provides evidence that
concurrent game structures with roles are an interesting semantics for
\acro{atl}, and should be investigated further.

Relating our work to ideas already present in the literature we find
it somewhat similar to the idea of exploiting symmetry in model
checking, as investigated by Sistla and Godefroid
\cite{SisGod0407-0}. However, our approach is different, since we look
at agent symmetries in \acro{atl} as the basis of a~new
semantics. When it comes to work related directly to strategic logics,
we find no similar ideas present, hence concluding that our approach
is indeed novel.

For future work we plan on investigating the homogeneous aspect of our
`roles' in more depth. We are currently working on a~derivative of
\acro{atl} with a~different language that will fully exploit the role
based semantics.

\paragraph{Acknowledgments:} We thank Pål Grønås Drange, Valentin
Goranko and Alessio Lomuscio for helpful comments. Piotr Kaźmierczak's
work was supported by the Research Council of Norway project 194521 (FORMGRID).

\bibliographystyle{eptcs}
\bibliography{references}

\end{document}